\documentclass[11pt,fleqn]{article}

\usepackage{amssymb,amsmath}
\usepackage{natbib}
\usepackage{vmargin,setspace}
\usepackage{graphicx,graphics,color,tikz}

\usepackage{libertine}
\usepackage[libertine]{newtxmath}

\definecolor{darkgreen}{rgb}{0.2,0.5,0.2}
\definecolor{darkred}{rgb}{0.6,0.1,0.1}

\usepackage[colorlinks=true, pdfstartview=FitV, linkcolor=darkgreen, citecolor=darkred, urlcolor=black]{hyperref}

\newcounter{llst}
\newenvironment{abet}{\begin{list}{\rm (\alph{llst})}{\usecounter{llst}
\setlength{\itemindent}{0em} \setlength{\leftmargin}{3em}
\setlength{\labelwidth}{2em} \setlength{\labelsep}{1em}}}{\end{list}}
\newenvironment{numm}{\begin{list}{\rm (\roman{llst})}{\usecounter{llst}
\setlength{\itemindent}{0em} \setlength{\leftmargin}{3.5em}
\setlength{\labelwidth}{2.5em} \setlength{\labelsep}{1em}}}{\end{list}}

\newtheorem{theorem}{Theorem}[section]

\newtheorem{axiom}[theorem]{Axiom}

\newtheorem{corollary}[theorem]{Corollary}

\newtheorem{definition}[theorem]{Definition}
\newtheorem{expl}[theorem]{Example}

\newtheorem{lemma}[theorem]{Lemma}

\newtheorem{dscrpt}[theorem]{Description}

\newenvironment{proof}[1][Proof]{\noindent \textbf{#1.} }{\hfill
\rule{0.5em}{0.5em}}

\setpapersize{A4}
\setmarginsrb{80pt}{40pt}{80pt}{60pt}{12pt}{10pt}{12pt}{30pt}

\begin{document}


\title{\textbf{The Core of an Economy with \\ an Endogenous Social Division of Labour}\thanks{I thank Marialaura Pesce, Dimitrios Diamantaras and Owen Sims for debates that guided me in developing this paper.}}

\author{Robert P.~Gilles\thanks{Economics Group, Management School, The Queen's University of Belfast, Riddel Hall, 185 Stranmillis Road, Belfast, BT9~5EE, UK. \textsf{Email: r.gilles@qub.ac.uk}} }

\date{\today}

\maketitle

\begin{abstract}
\singlespace
\noindent
This paper considers the core of a competitive market economy with an endogenous social division of labour. The theory is founded on the notion of a ``consumer-producer'', who consumes as well as produces commodities. First, we show that the Core of such an economy with an endogenous social division of labour can be founded on deviations of coalitions of arbitrary size, extending the seminal insights of Vind and Schmeidler for pure exchange economies. Furthermore, we establish the equivalence between the Core and the set of competitive equilibria for continuum economies with an endogenous social division of labour.

Our analysis also concludes that self-organisation in a social division of labour can be incorporated into the Edgeworthian barter process directly. This is formulated as a Core equivalence result stated for a Structured Core concept based on renegotiations among fully specialised economic agents, i.e., coalitions that use only fully developed internal divisions of labour.

Our approach bridges the gap between standard economies with social production and coalition production economies. Therefore, a more straightforward and natural interpretation of coalitional improvement and the Core can be developed than for coalition production economies.
\end{abstract}

\begin{description}
\singlespace
\item[Keywords:] Consumer-producer; Social division of labour; Competitive equilibrium; Core of an economy; Core equivalence.

\item[JEL classification:] D41, D51
\end{description}

\thispagestyle{empty}

\newpage

\setcounter{page}{1} \pagenumbering{arabic}

\section{Specialisation and the social division of labour}

One of the oldest ideas to explain economic wealth creation is that through a functional social division of labour \citep{Plato2007,Smith1776,Mandeville1714,Babbage1835}. Wealth creation in an economy with a social division of labour is founded on the interplay of two fundamental principles, namely that there are \emph{increasing returns to specialisation} and that full exploitation of these returns is possible through the principle of \emph{gains from trade}. These principles have already been explored in these early sources.

The principle that human capital is more productive if fully specialised in the execution of a limited set of activities or the production of a limited set of closely related goods has been proposed throughout the history of economic thought, but has not yet been explored fully in contemporary economic theory. This paper proposes formal concepts that describe these effects of Increasing Returns to Specialisation (IRSpec) that are weaker than the usually employed notion of Increasing Returns to Scale (IRS). This is formalised by the conception that the convex hull of the full specialisation production plans capture the complete production set.

We show in this paper that the property of Increasing Returns to Specialisation supports the two fundamental functions of a social division of labour, namely the generation of economic wealth and the allocation of that generated wealth. Indeed, specialised economic agents interact through a competitive price mechanism that acts as a coordination device in the process of economic wealth creation \emph{as well as} a mechanism to allocate the generated wealth through the social division of labour. To pursue this, we apply the framework developed in \citet{Yang1988} and \citet{YangNg1993}.\footnote{Alternative explorations to model the idea of wealth creation through a social division of labour have been presented in \citet{Young1928,Stigler1951,Blitch1983,Kim1989,Locay1990} and \citet{Nakahashi2014}. For an overview we also refer to \citet{Steinegger2010} and \citet{Sun2012}.} This approach represents an economic decision maker as a ``consumer-producer'', who is endowed with consumptive as well as productive abilities.\footnote{Technically this introduces a setting that is equivalent to the notion of a ``coalition production economy'' \citep{Hildenbrand1970,Hildenbrand1974} which is also explored in \citet{Suzuki1995}, \citet{Simone1997} and \citet{Toda2002}. Suzuki, De Simone and Toda focus mainly on the technical issue of the existence of an equilibrium rather than the role of specialisation in the endogenous formation of a social division of labour.} Yang presented this approach as a remedy to the formal (social) dichotomy of consumption and production that is at the centre of the Walrasian theory of a market economy \citep{ArrowDebreu1954,McKenzie1954,McKenzie1959}.

However, the concept of Increasing Returns to Specialisation remains relatively unexplored in formal models based on the notion of a consumer-producer. In particular, \citet{SunYangZhou2004} formalised a general equilibrium framework founded on the models of \citet{Yang1988,Yang2001,Yang2003} and \citet{YangNg1993}. In this framework, production sets of individual consumer-producers are typically bounded as well as non-convex. These authors restrict their model to one based on home-based production only and investigate the existence of competitive equilibrium under well-specified transaction costs. This formal model omits the explicit introduction of either Increasing Returns to Scale or Increasing Returns to Specialisation and, as a result, does not investigate the consequences of these fundamental hypotheses. Also, the focus on home-based production excludes the trade of intermediary inputs as part of social production processes, an essential feature of general equilibrium models of production economies. In the present paper we try to address these omissions.

We introduce the Core of an economy with an endogenous social division of labour in a rather straightforward fashion. As in coalition production economies \citep{Hildenbrand1968,Hildenbrand1974}, coalitions are assumed to be able to create wealth through the appropriate fine-tuning of the assigned production plans to its constituting members. The generated outputs are then allocated for consumption among the members of the coalition. This defines a standard Core concept in the setting of an economy with endogenous social division of labour.

We are able to show that non-Core allocations can be improved upon or ``blocked'' by coalitions of arbitrary size. This extends the well-known results of \citet{Schmeidler1972} and \citet{Vind1972} on blocking in pure exchange economies. Furthermore, we show that in a continuum economy there is an equivalence between the Core and the set of competitive equilibria, fully analogous to \citet{Aumann1964}. Our Core-equivalence result improves established insights from the literature on coalition production economies, in particular the results for pseudo- and quasi-equilibria shown in \citet{Hildenbrand1968} and \citet{Oddou1982}.


In this paper we explore the introduction of a formal description of Increasing Returns to Specialisation (IRSpec) at the level of the individual consumer-producer. Formally, a production set satisfies IRSpec if it is contained in the convex hull of production plans in which there is full specialisation in the production of a single good---subject to the free-disposal hypothesis. This property has been hinted at by \citet{Yang2001}, \citet{Yao2002} and \citet{DiamantarasGilles2004}, but has not been developed fully in that literature. We address this omission in the present paper as well as the companion paper on general competitive equilibrium in an economy with an endogenous social division of labour, \citet{Gilles2016}.

Finally, under IRSpec, any Core allocation can be improved upon by a coalition of fully specialised individuals. In particular, this coalition consists of non-negligible subcoalitions of consumer-producers that are specialised in the production each of the marketable commodities. This insight is related to the result of \citet{Grodal1972} for coalition production economies.

\subsection*{Outline of the paper}

In the second section of the paper we set up the model of a consumer-producer, in particular a generalised model of production already developed in \citet{Gilles2016}. This allows the introduction of the production of goods through an adaptable social division of labour based on the principles of increasing returns to specialisation and gains from trade.

Section 3 introduces the notion of an economy with an endogenous social division of labour and defines the notion of a competitive equilibrium. Subsequently, we introduce the notion of coalitional improvements and the concept of a Core allocation. We show that in this setting non-Core allocations can be blocked by coalitions of arbitrary, predetermined size, including arbitrarily large and arbitrarily small coalitions, extending the insights of \citet{Vind1972} and \citet{Schmeidler1972} to economies with production. We also extend the standard Core equivalence result that in continuum economies all Core allocations are competitive equilibrium allocations.

Finally, in Section 3 we explore an extension of the Core, capturing that only coalitions in which members are fully specialised are able to improve upon a given allocation. We show that these $\delta$-Core allocations are fully equivalent to the Core allocations as well as the competitive equilibrium allocations in a continuum economy in which all agents have access to the same productive specialisations.

All proofs of the main results are collected in several appendices.

\section{Economies with a social division of labour}

Throughout, we consider that there are $\ell \geqslant 2$ marketable commodities. The \emph{commodity space} is therefore represented as the $\ell$-dimensional Euclidean space $\mathbb{R}^{\ell}$ and the \emph{consumption space} is its nonnegative orthant $\mathbb{R}^{\ell}_+$. The commodity space represents all bundles of tradable goods in this economy. In particular, for $k = 1, \ldots ,\ell$ we denote by $e_k = (0, \ldots ,0, 1 ,0, \ldots ,0)$ the $k$-th unit bundle in $\mathbb{R}^{\ell}_+$ and by $e = (1, \ldots ,1)$ the bundle consisting of one unit of each good.\footnote{Throughout, we employ the vector inequality notation that $x \geqslant x'$ if $x_k \geqslant x'_k$ for all commodities $k = 1, \ldots ,\ell$; $x > x'$ if $x \geqslant x'$ and $x \neq x'$; and $x \gg x'$ if $x_k > x'_k$ for all commodities $k = 1, \ldots ,\ell$.}

In this paper we consider an economy with a diversified production sector, which is founded on the hypothesis that all agents are participating directly in the production as well as the consumption of goods. Hence, the production in this economy is founded on an endogenous social division of labour that emanates from decisions of all individual economic agents. Consequently, economic agents are modelled as \emph{consumer-producers}---a notion formally conceived by \citet{Yang1988}. 

\subsection{Introducing consumer-producers}

A consumer-producer is formally introduced as a pair $(u, \mathcal P )$, where $u \colon \mathbb{R}^{\ell}_+ \to \mathbb{R}$ is a utility function representing the agent's consumptive preferences and $\mathcal{P} \subset \mathbb{R}^{\ell}$ is the agent's production set consisting of production plans. Normally we can write every production plan $y \in \mathcal P$ as $y = y^+ - y^-$, where $y^+ = y \vee 0$ is the corresponding vector of output quantities and $y^- = (-y) \vee 0$ is the corresponding vector of input quantities.

We also assume that there can be non-marketable inputs in the production process described by $\mathcal P$---such as the agent's labour time and land resources---that are not explicitly modelled. We facilitate the possibility that all outputs are generated using non-marketable inputs only. This is referred to as ``home'' production. The approach to the modelling of production introduced here is similar to the one set out for coalition production economies in \citet[Section 2.1]{Sondermann1974}.\footnote{In a coalition production economy, productive abilities are assigned to coalitions of economic agents. For further discussion of coalition production economies I refer to \citet{Hildenbrand1970,Sondermann1974,Oddou1982} and \citet{Basile1993}.}

In this paper, the concept of a consumer-producer is introduced in a more general fashion than considered in \citet{Yang2001}, \citet{Yao2002}, \citet{SunYangZhou2004} and \citet{DiamantarasGilles2004}. In those papers, a consumer-producer is endowed with one unit of non-marketable labour time which can be allocated to either consumption as leisure or as production of any of the considered marketable commodities. Here, instead, we dispense of the explicit introduction of labour time and simplify the framework by considering production sets similar to the ones used in standard general equilibrium theory with inputs and outputs.

With regard to our approach to production, in the production set $\mathcal{P}_a$ of consumer-producer $a \in A$, we distinguish two types of inputs: \emph{Marketable inputs}, which are part of the $\ell$ commodities introduced above, and \emph{non-marketable inputs}, which are assumed to be outside the realm of the commodity space. 

\paragraph{Regularity assumptions.}

The next assumption brings together the regularity properties that one expects to be satisfied by the basic descriptors of the consumptive and productive abilities of a consumer-producer. 
\begin{axiom} \label{ax:basic}
Throughout this paper, we assume that every agent is a consumer-producer represented by $(u, \mathcal{P})$, satisfying the following properties:
\begin{numm}
\item The utility function $u \colon \mathbb{R}^{\ell}_+ \to \mathbb{R}$ is \emph{continuous} on the consumption space $\mathbb{R}^{\ell}_+$ and \emph{monotone} in the sense that for all consumption bundles $x,x' \in\mathbb{R}^{\ell}_+$ with $x \gg x'$ it holds that $u(x) > u(x')$.

\item The production set $\mathcal{P} \subset \mathbb{R}^{\ell}$ is \emph{regular} in the sense that $\mathcal{P}$ is a closed set that is bounded from above, $0 \in \mathcal{P}$ and it is comprehensive in the sense that
\begin{equation}
\mathcal{P} - \mathbb{R}^{\ell}_+ \equiv \left\{ \, y-z \, \left| \, y \in \mathcal{P} \mbox{ and } z \geqslant 0 \, \right. \right\} \subset \mathcal{P} .
\end{equation}
\item The production set $\mathcal{P} \subset \mathbb{R}^{\ell}$ is \emph{non-trivial} in the sense that
\begin{equation}
	\mathrm{Conv} \, \left( \mathcal{P} \right) \cap \mathbb R^{\ell}_{+} \setminus \{ 0 \} \neq \varnothing ,
\end{equation}
i.e., there exists some $z > 0$ with $z \in \mathrm{Conv} \, \mathcal P$.\footnote{Here, we denote by $\mathrm{Conv} \, (X) = \{ \lambda x + (1- \lambda )y \mid x,y \in X$ and $0 \leqslant \lambda \leqslant 1 \}$ the convex hull of the set $X$.}
\end{numm}
\end{axiom}
Axiom \ref{ax:basic}(i) imposes that we only consider ``regular'' preferences in this economy. The imposed properties of continuity and (weak) monotonicity are standard assumptions in the general equilibrium literature. For certain results, we impose the additional assumption that preferences are \emph{strictly monotone} in the sense that $u (x) > u (x')$ for all $x>x'$.

Furthermore, Axiom \ref{ax:basic}(ii) imposes that we only consider production sets that satisfy certain standard regularity properties. In particular, a regular production set satisfies the properties that one has the ability to cease production altogether and the assumption of free disposal in production. Both properties are used throughout general equilibrium theory.

Also, it is natural to assume that due to their size individual consumer-producers can only manage bounded production processes and are not able to grow their operations arbitrarily. This is expressed through the property that there is an upper bound on the individual's production set.

Note that regularity does not include convexity, allowing production to exhibit non-convexities, in particular increasing returns to scale and specialisation, introducing constructions of production sets developed in the existing literature on economies with an endogenous social division of labour \citep{Yang1988,Yang2001,SunYangZhou2004,Gilles2016}.

Finally, Axiom \ref{ax:basic}(iii) requires that production is non-trivial and that economic agents can achieve a meaningful productive output. The stated property implies that economic agents can produce certain outputs without using too much inputs. This implies that all economic agents are assumed to have meaningful productive abilities and can contribute to the social economy productively.

\subsection{Defining an economy based on a social division of labour}

The set of agents is denoted by $A$ and a typical economic agent is $a \in A$. In this paper, $A$ is either a finite set or a continuum endowed with a complete, atomless probability measure to determine the size of coalitions in the population.\footnote{There also is the possibility to consider mixed cases as well in which there is a continuum of negligible agents and countably many non-negligible agents, represented by atoms \citep{Shitovitz1973,Shitovitz1982}. We do not address these cases in the theory developed in this paper.} Formally, we let $A$ be an arbitrary set of economic consumer-producers represented by $(u_a,\mathcal{P}_a)$. Let $\Sigma \subset 2^A$ be a well-chosen $\sigma$-algebra of measurable coalitions in $A$ and let $\mu \colon \Sigma \to [0,1]$ be a complete probability measure on $(A, \Sigma )$.

We distinguish two standard cases: (1) $A = \{ 1, \ldots ,n \}$ is a countable set, $\Sigma =2^A$ is the class of all subsets of $A$\footnote{On a non-finite, countable agent set $A$, $\Sigma$ can also be interpreted here as the $\sigma$-algebra generated by the finite subsets of $A$.} and $\mu$ is the normalised counting measure on $\Sigma$. (2) $A$ is an uncountably infinite set, $\Sigma$ is a $\sigma$-algebra of measurable subsets of $A$ and $\mu$ is an atomless probability measure on $\Sigma$. In the remainder of this paper case (1) is denoted as a ``finite'' economy. The second case (2) is denoted as a ``continuum'' economy.

We now develop the model of an economy in which all agents are represented as consumer-producers and simultaneously make decisions regarding consumption and production under a given vector of competitive market prices. It is implicitly assumed that all consumer-producers are powerless to influence the price level through their consumption and/or production decisions. Hence, the model is a direct extension of the standard model of an economy with (social) production.
\begin{definition} \label{def:economy}
An \textbf{economy} with $\ell$ commodities is a triple $\mathbb{E} = \langle \, (A, \Sigma ,\mu ), u, \mathcal{P} \rangle$ where
\begin{numm}
\item $(A, \Sigma ,\mu )$ is a complete probability space of consumer-producers, represented by $(u_a , \mathcal P_a )$ for every $a \in A$;

\item The utility function $u ( \cdot ,x) \colon A \to \mathbb{R}$ is measurable on the probability space $(A, \Sigma ,\mu )$;

\item The correspondence of production sets $\mathcal{P} \colon A \to 2^{\mathbb{R}^{\ell}_+}$ has a measurable graph on the probability space $(A, \Sigma ,\mu )$ such that the correspondence $\mathcal P$ is integrably bounded from above\footnote{Hence, we impose that there exists some integrable function $\overline g \colon A \to \mathbb R^{\ell}$ such that for almost every agent $a \in A$ and every production plan $y \in \mathcal P_A \colon y \leqslant \overline g (a)$.} and;

\item For every agent $a \in A$ the pair $(u_a , \mathcal{P}_a )$ satisfies Axiom \ref{ax:basic}.
\end{numm}
An economy $\, \mathbb{E} = \langle \, (A, \Sigma ,\mu ), u, \mathcal{P} \rangle$ is a \textbf{continuum economy} if the complete probability space of consumer-producers $(A, \Sigma ,\mu )$ is atomless.
\end{definition}
An \emph{allocation} in an economy $\mathbb{E} = \langle \, (A, \Sigma ,\mu ), u, \mathcal{P} \rangle$ is a pair $(f,g)$ where $f \colon A \to \mathbb{R}^{\ell}_+$ is an integrable assignment of final consumption bundles and $g$ is an integrable selection of the production correspondence $\mathcal{P}$. Thus, an allocation $(f,g)$ consists of a consumption plan $f$ and a social production plan $g$. 

An allocation $(f,g)$ is now \emph{feasible} if all allocated consumption bundles are covered by the produced quantities of all commodities, i.e.,
\begin{equation}
\int f \, d\mu \leqslant \int g \, d\mu .
\end{equation}

\section{Competitive equilibria and the Core}

\paragraph{The consumer-producer problem.}

An individual consumer-producer $a \in A$ operating in a system of perfectly competitive markets faces a vector $p \in \mathbb{R}^{\ell}_+ \setminus \{ 0 \}$ of prices for the $\ell$ marketable commodities. Consequently, she optimises her utility over the set of the feasible consumption-production plans with the goal to maximise consumptive satisfaction. This optimisation problem is denoted as the ``consumer-producer problem'':

\medskip\noindent
Every consumer-producer $a \in A$ aims to select a consumption plan $\hat{x} \in \mathbb{R}^{\ell}_+$ and production plan $\hat{y} \in \mathbb{R}^{\ell}$ such that $(\hat{x} , \hat{y})$ solves
\begin{equation} \label{CP-Prob}
\max \, \left\{ u_a (x) \, \left| \, x \in \mathbb{R}^{\ell}_+ \, , \ y \in \mathcal{P}_a \, \mbox{ and } \, p \cdot x \leqslant p \cdot y \, \right. \right\} .
\end{equation}
In this statement, $\hat{x} \in \mathbb{R}^{\ell}_+$ denotes the \textbf{final consumption bundle}, while $\hat{y} \in \mathbb{R}^{\ell}$ is referred to as the \textbf{production plan} executed by agent $a$. Now, $t = \hat{y}- \hat{x} \in \mathbb{R}^{\ell}$ is the vector of net trades submitted to the market.

\medskip\noindent
We first establish that the consumer-producer problem stated here has actually a solution under the regularity conditions imposed.
\begin{lemma} \label{thm:CP-exist}
\textbf{\emph{(Solution to the Consumer-Producer Problem)}} \\
Let $a \in A$ and let $(u_a, \mathcal P_a)$ satisfy regularity Axiom \ref{ax:basic}. Then the consumer-producer problem for $a \in A$ stated in (\ref{CP-Prob}) has a solution for any strictly positive price vector $p \gg 0$.
\end{lemma}
A proof of this existence result follows immediately from the weak dichotomy between the consumption and production decisions stated in \citet{Gilles2016}.

\paragraph{Competitive equilibria.}

An equilibrium in an economy can now be introduced as a feasible allocation that is supported by a price vector for which nearly all consumer-producers select optimal consumption and production plans.
\begin{definition}
Consider an economy $\mathbb{E} = \langle \, (A, \Sigma ,\mu ) , u , \mathcal{P} \rangle$.
\\
An allocation $(f,g)$ is a \textbf{competitive equilibrium} in  $\mathbb{E}$ if $(f,g)$ is feasible and there exists some strictly positive price vector $p \gg 0$ such that $\int f \, d\mu = \int g \, d\mu$ and for almost every agent $a \in A$, the consumption-production plan $(f(a),g(a))$ solve $a$'s consumer-producer problem (\ref{CP-Prob}).
\\
The price vector $p$ is denoted as an\textbf{ equilibrium price} for the equilibrium allocation $(f,g)$.
\\
The collection of all equilibrium allocations in the economy $\mathbb E$ is denoted by
\begin{equation}
	\mathcal W (E) = \{ (f,g) \mid (f,g) \mbox{ is a competitive equilibrium in } \mathbb E \} .
\end{equation}
\end{definition}
We restrict our focus on competitive equilibria with strictly positive equilibrium prices only. This restriction is more plausible in the setting of an economy with agent-based production as considered here. This restriction also strengthens the Core equivalence result developed subsequently in this paper, since it reduces the size of the set of considered equilibria in the economy.

The existence of competitive equilibria in these economies with an endogenous social division of labour is addressed in \citet{Gilles2016}, in particular Section 3 and Theorem 3.6. There it is shown that to show existence of meaningful competitive equilibria, the production sets of the consumer-producers in the economy have to satisfy properties describing Increasing Returns to Specialisation. We refer to Section 4 of this paper for an elaborate discussion.

\paragraph{A Core concept.}

The Core of a continuum economy with production was seminally set out in \citet{Hildenbrand1968}. The model proposed was that of a coalition production economy in which production was assigned to coalitions rather than separate social production organisations or ``producers''. In our notion of an economy with an endogenous social division of labour, production is assigned to the consumer-producers individually, which results in a natural formulation of coalitional improvement in the Edgeworthian bargaining process. Consequently, we are able to formulate the concept of an Edgeworthian Core allocation in a rather straightforward fashion.
\begin{definition}
Let $\, \mathbb{E} = \langle \, (A, \Sigma ,\mu ), u, \mathcal{P} \rangle$ be some economy. A coalition $S \in \Sigma$ \textbf{is able to improve upon} an allocation $(f,g)$ in $\mathbb{E}$ if $S$ is non-negligible in the sense that $\mu (S)>0$ and there exists a coalitional allocation $(f',g') \colon S \to \mathbb R^{\ell}_+ \times \mathbb R^{\ell}$ such that
\begin{numm}
\item $u_a \left( f'(a) \, \right) > u_a (f(a))$ for almost every agent $a \in S$, and

\item $\int_S f' \, d\mu \leqslant \int_S g' \, d\mu$.
\end{numm}
A feasible allocation $(f,g)$ in $\mathbb{E}$ is a \textbf{Core allocation} if there is no coalition $S \in \Sigma$ that can improve upon it. We denote by
\begin{equation}
	\mathcal C (\mathbb E) = \{ (f,g) \mid (f,g) \mbox{ is a Core allocation} \, \}
\end{equation}
the \textbf{Core} of the economy $\mathbb E$.
\end{definition}
A coalition that can improve upon a given allocation is also known as a \emph{blocking coalition}.

\subsection{An extension of the Vind-Schmeidler Theorem.}

One of the classical questions in mathematical economics has been what type of blocking coalitions there are. One particular question is about the size of the blocking coalition. This was pursued in \citet{Schmeidler1972,Grodal1972,Vind1972} and \citet{GWY2007}. These results can quite naturally be transferred to the setting of an economy with an endogenous social division of labour.

The next result addresses the size of blocking coalitions there emerge in a continuum economy. It extends the insights from \citet{Schmeidler1972} and \citet{Vind1972} to the realm of continuum economies with an endogenous social division of labour. In particular, in the context of our notion of an economy, the concept of a consumer-producer incorporates the property that production is completely scalable as well through the adaptation of the social division of labour. This property carries over to blocking coalitions in the sense that improving on a proposed allocation can be done through a set of completely specialised coalitions of consumer-producers.
\begin{theorem} \label{thm:VindSchmeidler}
	Let $\, \mathbb{E} = \langle \, (A, \Sigma ,\mu ), u, \mathcal{P} \rangle$ be some continuum economy, in which $(A, \Sigma ,\mu )$ is a complete atomless probability space for which the properties stated in Definition \ref{def:economy} hold and all preferences are strictly monotone. Let $0 < \delta <1$ be any real number. Then for any non-Core allocation $(f,g)$ in $\mathbb{E}$, there exists a blocking coalition $S \in \Sigma$ with $\mu (S) = \delta$.
\end{theorem}
For a proof of Theorem \ref{thm:VindSchmeidler} we refer to Appendix \ref{App-VindSchmeidler} of this paper.

\subsection{Core equivalence}

One of the classical results from the Edgeworthian approach to trade has been the Core equivalence theorem, seminally established for pure exchange economies by \citet{Aumann1964} and for coalition production economies by \citet{Hildenbrand1968}. Here, we are able to establish full Core equivalence for a continuum economy with an endogenous social division of labour, supplementing the results stated in \citet{Hildenbrand1968,Hildenbrand1974} for coalition production economies. We are able to establish the equivalence of the Core and the set of competitive equilibria supported by strictly positive prices for economies with an endogenous social division of labour under the accepted monotonicity hypotheses.
\begin{theorem} \label{thm:CE}
\textbf{\emph{(Core Equivalence)}} \\
Let $\, \mathbb{E} = \langle \, (A, \Sigma ,\mu ), u, \mathcal{P} \rangle$ be some economy that satisfies the properties stated in Definition \ref{def:economy}.
\begin{abet}
\item Every competitive equilibrium in $\mathbb{E}$ is a Core allocation.

\item If $(A, \Sigma ,\mu )$ is an atomless probability space and all utility functions $u_a$, $a \in A$, are strictly monotone, then every Core allocation in $\mathbb{E}$ can be supported as a competitive equilibrium with a strictly positive equilibrium price vector, i.e., $\mathcal C (\mathbb E) = \mathcal W (\mathbb E)$.
\end{abet}
\end{theorem}
A proof of Theorem \ref{thm:CE} is provided in Appendix \ref{App-D} of this paper.

\medskip\noindent
The Core equivalence result establishes that our framework is a natural representation of modelling production in a continuum economy without relying on weaker concepts such as pseudo- and quasi-equilibria.

\section{The Core and the social division of labour}

\subsection{Increasing returns to specialisation}

The generation of economic wealth in a social division of labour is founded on the presence of increasing returns to specialisation: Specialising in the production of a single good leads to disproportional increases in the marginal product of that particular good. In the framework developed here, this fundamental hypothesis is introduced in two forms.
\begin{definition} \label{def:IRSpec}
Consider a production set $\mathcal{P} \subset \mathbb{R}^{\ell}$ satisfying the properties stated in Axiom \ref{ax:basic}.
\begin{description}
\item[Full specialisation:] Let $k \in \{ 1, \ldots , \ell \}$ be some commodity. A production plan $z^k \in \mathcal P$ is a \textbf{full specialisation production plan} for commodity $k $ if there exists some positive output quantity $Q^k >0$ and some input vector $y^k \in \mathbb{R}^{\ell}_+$ with $y^k_k =0$ such that $z^k = Q^k e_k - y^k$.
\item[IRSpec:] The production set $\, \mathcal{P}$ exhibits \textbf{Increasing Returns to Specialisation} if for every commodity $k \in \{ 1, \ldots , \ell \}$ there exists some full specialisation production plan $z^k = Q^k e_k - y^k \in \mathcal P$ such that \footnote{Here, we use the notational convention that $\mathrm{Conv} \, S = \{ \lambda x + (1- \lambda )y \mid x,y \in S$ and $0 \leqslant \lambda \leqslant 1 \}$ denotes the convex hull of the set $S$.}
\begin{equation}
\label{eq:WIRSpec}
\mathcal Q \subset \mathcal{P} \subset \mathrm{Conv} \, \mathcal{Q} - \mathbb{R}^{\ell}_+
\end{equation}
where $\mathcal{Q} = \left\{ \, z^1 , \ldots , z^{\ell} \, \right\}$ is the set of these $\ell$ full specialisation production plans.
\end{description}
\end{definition}
The definition above reflects the traditional idea underpinning the social division of labour \citep[Chapter 1]{Smith1776} that specialising in the production of a single good increases one's productivity. If a production set exhibits IRSpec, the production set is simply non-convex with the maximal production bundles achieved along the $\ell$ full specialisation production plans. Here, $Q^k >0$ denotes the quantity of good $k$ that the consumer-producer can produce under full specialisation, requiring input vector $y^k \geqslant 0$. This notion was seminally introduced in \citet{Gilles2016}, where two variations of Increasing Returns to Specialisation are considered. The property introduced here is denoted as Weakly Increasing Returns to Specialisation (WIRSpec) in \citet{Gilles2016}.\footnote{\citet{Gilles2016} also introduces the related notion of Strictly Increasing Returns to Specialisation (SIRSpec), which additionally imposes that $\mathcal P \cap ( \mathrm{Conv} \mathcal Q ) = \mathcal Q$. This imposes that there are strict productivity increases from specialising in the production of a single output. This implies that income maximisers are exactly those full specialisation production plans.}

Obviously, WIRSpec implies that the production set $\mathcal{P}$ is strictly bounded. We also remark that, if $\mathcal{P}$ is home-based, then $\mathcal{P}$ satisfies WIRSpec if there exist output quantities $Q^k >0$ for all $k \in \{ 1, \ldots ,\ell \}$ with $\mathcal{P} \subset \mathrm{Conv} \, \mathcal{Q} - \mathbb{R}^{\ell}_+$ where $\mathcal{Q} = \left\{ Q^1 e_1 , \ldots , Q^{\ell} e_{\ell} \right\}$.

For an elaborate discussion of these concepts and their consequences we also refer to \citet{Gilles2016}.

\subsection{A Structured Core concept}

For the subsequent discussion of the role of the social division of labour in the Edgeworthian trade processes underlying the Core, we need to introduce the notion that the production of goods in an economy is founded on a given set of professions or specialisations. In particular, we focus on the ability of professional associations or \emph{guilds} to arise that can exercise the power to block a certain allocation. This underlies the following definition of a modification of the Core concept.
\begin{definition}
	Let $\mathbb{E} = \langle \, (A, \Sigma ,\mu ), u, \mathcal{P} \rangle$ be some economy satisfying the properties of Definition \ref{def:economy} such that almost all agents are endowed with a production set that satisfies IRSpec. \\ Denote by $\mathcal{Q}_a \subset \mathbb{R}^{\ell}$ the set of $\ell$ full specialisation production plans for agent $a \in A$ with the property that $\mathcal{Q}_a \subset \mathcal{P}_a \subset \mathrm{Conv} \, \mathcal{Q}_a - \mathbb{R}^{\ell}_+$, as introduced in Definition \ref{def:IRSpec}.
	\begin{numm}
		\item Let $S \in \Sigma$ with $\mu (S) >0$ be a non-negligible coalition. A production plan $g \colon S \to \mathbb R^{\ell}$\textbf{structures} the coalition $S$ through an internal social division of labour if for every $k \in \{ 1, \ldots ,\ell \}$ there exists a coalition $S_k \in \Sigma$ with $\mu (S_k) >0$ such that
		\begin{enumerate}
			\item $\{ S_1 , \ldots , S_{\ell} \}$ partitions $S$, i.e., $S = \cup^{\ell}_{k=1} S_k$ and $S_m \cap S_{m'} = \varnothing$ for all $m \neq m'$ and;
			\item every agent $a \in S_k$ is specialised in the production of commodity $k$ only, in the sense that $g(a) = Q^k (a) e_k - y^k(a) \in \mathcal Q_a \subset \mathcal P_a$ is a full specialisation production plan for commodity $k$ for agent $a \in S_k$.
		\end{enumerate} 
		\item The allocation $(f,g)$ is a \textbf{Structured Core allocation} if there is no non-negligible coalition $S \in \Sigma$ with $\mu (S) >0$ that can improve upon $(f,g)$ through an allocation $(f',g') \colon S \to \mathbb R^{\ell}_+ \times \mathbb R^{\ell}$ such that the production plan $g'$ structures the coalition $S$ through an internal social division of labour.
	\end{numm}
	The collection of all Structured Core allocations in $\mathbb E$ is denoted by $\mathcal C^S ( \mathbb E)$.
\end{definition}
A coalition that organises itself through an internal social division of labour based on full specialisation production plans only can be interpreted as an alliance between $\ell$ different professional guilds, which members specialise in the production of only one good. The Structured Core now collects exactly those allocations that cannot be blocked through such alliances. It means that trade only occurs between fully specialised economic agents in the prevailing social division of labour. This imposes restrictions on blocking, typically enlarging the Core considerably, i.e., $\mathcal C ( \mathbb E) \subset \mathcal C^S ( \mathbb E)$.

The next theorem asserts that there is an equivalence between Core and Structured Core allocations in continuum economies in which the production technologies satisfy the increasing returns to specialisation (IRSpec) property, subject to standard regularity conditions on the preferences of the economic agents.

\begin{theorem} \textbf{\emph{(Structured Core Equivalence)}} \label{thm:Blocking} \\
Let $\, \mathbb{E} = \langle \, (A, \Sigma ,\mu ), u, \mathcal{P} \rangle$ be a continuum economy, for which the assumptions formulated in Definition \ref{def:economy} hold and almost all agents are endowed with a production set that satisfies IRSpec. \\ Then, if all utility functions $u_a$, $a \in A$, are strictly monotone, every non-Core allocation can be improved upon by a non-negligible coalition $S \in \Sigma$ with an internal social division of labour, i.e.,
\[
\mathcal C^S ( \mathbb E) = \mathcal C (\mathbb E) = \mathcal W (\mathbb E).
\]
\end{theorem}
A proof of Theorem \ref{thm:Blocking} is provided in Appendix \ref{App-Blocking}.

\bigskip\noindent
Using the proof of Theorem \ref{thm:VindSchmeidler}, a simple modification of the proof of Theorem \ref{thm:Blocking} implies the following corollary.
\begin{corollary}
Let $\, \mathbb{E} = \langle \, (A, \Sigma ,\mu ), u, \mathcal{P} \rangle$ be a continuum economy, for which the assumptions formulated in Definition \ref{def:economy} hold and almost all agents are endowed with a production set that satisfies IRSpec. Let $0 \leqslant \delta \leqslant 1$. \\ Then, if all utility functions $u_a$, $a \in A$, are strictly monotone, every non-Core allocation can be improved upon by a non-negligible coalition $S \in \Sigma$ with $\mu (S) = \delta$ that is structured through an internal social division of labour.
\end{corollary}

\subsection{Structured Equilibria under IRSpec}

Under the Increasing Returns to Specialisation (IRSpec) property of the productive abilities of the consumer-producers in an economy, \citet{Gilles2016} showed that the selected production plans under any positive price vector are fully specialised \citep[Theorem 2.7]{Gilles2016}. This gives rise to introducing a class of competitive equilibria that incorporate this property.
\begin{definition}
Let $\, \mathbb{E} = \langle \, (A, \Sigma ,\mu ), u, \mathcal{P} \rangle$ be an economy, for which the assumptions formulated in Definition \ref{def:economy} hold and almost all agents $a \in A$ are endowed with a production set $\mathcal P_a$ that satisfies IRSpec for the set of $\ell$ full specialisation production plans $\mathcal Q_a = \{ z^1_a , \ldots , z^{\ell}_a \}$.
\\
An allocation $(f,g)$ is a \textbf{structured equilibrium} in $\mathbb E$ if there exists a positive price vector $p \gg 0$ such that $\int f \, d\mu = \int g \, d\mu$ and for almost every agent $a \in A$, the consumption-production plan $(f(a),g(a))$ solve $a$'s consumer-producer problem (\ref{CP-Prob}) with $g(a) \in \mathcal Q_a$.
\\
The class of structured equilibria of $\, \mathbb E$ is denoted by $\mathcal W^S ( \mathbb E ) \subset \mathcal W (\mathbb E )$.
\end{definition}
The notion of a structured equilibrium refers to the endogenous organisation of the economy through a social division of labour through a price mechanism as seminally alluded to by \citet{Smith1759,Smith1776} as his ``invisible hand'' mechanism. Here I am particularly interested in the question whether interactions within coalitions that are internally organised through an appropriately chosen social division of labour lead to outcomes that are exactly those identified structured equilibria. Hence, do interactions in structured coalitions result in global outcomes that are guided through a price-based invisible hand to a global social division of labour?

Our insights indeed confirm that this is the case. Formally, this can be stated as a corollary from the previously shown insights. In particular, application of Theorem 2.7 of \citet{Gilles2016} and the Structured Core Equivalence Theorem \ref{thm:Blocking} results in the following corollary:
\begin{corollary}
Let $\, \mathbb{E} = \langle \, (A, \Sigma ,\mu ), u, \mathcal{P} \rangle$ be an economy, for which the assumptions formulated in Definition \ref{def:economy} hold and almost all agents $a \in A$ are endowed with a production set $\mathcal P_a$ that satisfies IRSpec for the set of $\ell$ full specialisation production plans $\mathcal Q_a$.
\\
Then, if all utility functions $u_a$, $a \in A$, are strictly monotone, then every structured core allocation can be supported as a structured equilibrium, i.e.,
\[
\mathcal C^S ( \mathbb E) = \mathcal C (\mathbb E) = \mathcal W (\mathbb E) = \mathcal W^S (\mathbb E).
\]
\end{corollary}

\newpage
\singlespace
\bibliographystyle{econometrica}
\bibliography{RPDB}

\newpage
\appendix

\section*{Appendices: Proofs of the main theorems}

\section{Proof of Theorem \ref{thm:VindSchmeidler}} \label{App-VindSchmeidler}

Consider a continuum economy $\, \mathbb{E} = \langle \, (A, \Sigma ,\mu ), u, \mathcal{P} \rangle$ such that $(A, \Sigma ,\mu )$ is atomless. Let $(f,g)$ be some non-Core allocation in $\mathbb{E}$. Hence, there exists some $S \in \Sigma$ with $\mu (S)>0$ and a coalitional allocation $(f',g')$ with $f' \colon S \to \mathbb{R}^{\ell}_+$ and $g' \colon S \to \mathbb{R}^{\ell}$ such that
\begin{numm}
\item $u_a(f'(a)) > u_a(f(a)) \equiv \bar{u}_a$ for every $a \in S$;

\item $g'(a) \in \mathcal{P}_a$ for every $a \in S$, and;

\item $\int_S f' \, d\mu = \int_S g' \, d\mu$.
\end{numm}
Let $0< \delta <1$. We now consider three cases:

\subsubsection*{\maltese \ \ $\mu(S) = \delta$}

In this case the assertion of Theorem \ref{thm:VindSchmeidler} is satisfied for $S$ itself.

\subsubsection*{\maltese \ \ $\mu(S) > \delta$}

A proof can easily be constructed on the argument introduced in \citet{Schmeidler1972}. Indeed, introduce a multi-dimensional measure $\nu \colon \Sigma \to \mathbb{R}^{\ell+1}_+$ on $(A, \Sigma ,\mu )$ restricted to $S$ by
\begin{equation}
\nu (T) = \left( \int_T (f'-g')\, d\mu \, , \, \mu (T) \, \right) \, \in \mathbb{R}^{\ell+1}_+ \quad \mbox{for any } T \subset S
\end{equation}
Now by Liapunov's Convexity Theorem, e.g., Theorem 3 in \citet[page 62]{Hildenbrand1974}, it follows that $\nu$ results in a convex image.
\\
Obviously $\nu (\varnothing )= (0, \ldots ,0,0)$ and $\nu (S) = (0, \ldots ,0, \mu (S))$. Hence, there has to exist some $T \subset S$ such that $\nu (T) = (0, \ldots ,0, \delta )$. Clearly the coalition $T$ now improves upon $(f,g)$ through $(f',g')$ such that $\mu (T) = \delta$, showing the assertion.

\subsubsection*{\maltese \ \ $\mu(S) < \delta$}

We modify the proof seminally devised by \citet{Vind1972}. Let $0< \varepsilon <1$ be arbitrarily chosen.
\\
Since the restricted probability space $(S, \Sigma_S ,\mu )$ is atomless as well, we know by Liapunov's Convexity Theorem that $\int_S \mathcal{P}_a \, d\mu (a) \subset \mathbb{R}^{\ell}_+$ is a convex set. Hence, there exists an integrable selection $\hat{g} \colon S \to \mathbb{R}^{\ell}_+$ such that $\hat{g}(a) \in \mathcal{P}_a$ for $a \in S$ and
\begin{equation}
\label{eq:ProdVind}
\int_S \hat{g} \, d\mu = \varepsilon \int_S g' \, d\mu + (1- \varepsilon ) \int_S g \, d\mu
\end{equation}
Furthermore, the continuity (Axiom \ref{ax:basic}(i)) and measurability of all preferences implies that there exists $f'' \colon S \to \mathbb{R}^{\ell}_+$ with
\begin{gather*}
u_a \left( f'' (a) \, \right) > \bar{u}_a \quad \mbox{for every } a \in S \\
\int_S f'' \, d\mu = \int_S g' \, d\mu -y \quad \mbox{for some } y>0 .
\end{gather*}
Since the correspondence $\{ x \in \mathbb{R}^{\ell}_+ \mid u_a(x) > \bar{u}_a \}$ is measurable on $S$, by Liapunov's Convexity Theorem it follows that
\[
\int_S \left\{ \left. \, x \in \mathbb{R}^{\ell}_+ \, \right| \, u_a(x) > \bar{u}_a \, \right\} \, d\mu (a) \subset \mathbb{R}^{\ell}_+
\]
is a convex set. Hence, since $\int_S f'' \, d\mu$ is in this integral and $\int_S f \, d\mu$ is on its boundary, there exists some integrable selection $\hat{f} \colon S \to \mathbb{R}^{\ell}_+$ such that
\begin{gather}
u_a \left( \, \hat{f} (a) \, \right) > \bar{u}_a \quad \mbox{for every } a \in S \\
\int_S \hat{f} \, d\mu = \varepsilon \int_S f'' \, d\mu + (1- \varepsilon ) \int_S f \, d\mu
\end{gather}
Since $\mu$ is atomless, we can now take $T \subset A \setminus S$ with $T \in \Sigma$, $\mu (T) = (1- \varepsilon ) \mu (A \setminus S) >0$ and
\begin{equation}
\int_T (f-g) \, d\mu = (1- \varepsilon ) \int_{A \setminus S} (f-g)\, d\mu .
\end{equation}
Now consider the coalition $S \cup T \in \Sigma$ and define the allocation $\left( \, \tilde{f}, \tilde{g} \, \right)$ by
\[
\tilde{f} (a) = \left\{
\begin{array}{ll}
\hat{f} (a) & a \in S \\
f(a) + \frac{\varepsilon}{\mu (T)} y & a \in T
\end{array}
\right.
\quad \mbox{and} \quad  \tilde{g} (a) = \left\{
\begin{array}{ll}
\hat{g} (a) & a \in S \\
g(a) & a \in T
\end{array}
\right. .
\]
Now it follows easily from the above that for $a\in S$
\[
u_a \left( \tilde{f} (a) \right) = u_a \left( \hat{f} (a) \right) > \bar{u}_a
\]
and, due to strict monotonicity of $u_a$, for every $a \in T$
\[
u_a \left( \tilde{f} (a) \right) = u_a \left( f (a) + \tfrac{\varepsilon}{\mu (T)} y \right) > \bar{u}_a .
\]
Moreover, we establish that
\begin{align*}
\int_{S \cup T} \tilde{f} \, d\mu & = \int_S \hat{f} \, d\mu + \int_T f \, d\mu +\varepsilon y = \varepsilon \int_S f'' \, d\mu + (1- \varepsilon ) \int_S f\, d\mu + \int_T f \, d\mu +\varepsilon y = \\
& = \varepsilon \int_S g' \, d\mu - \varepsilon y + (1- \varepsilon ) \int_S f\, d\mu + \int_T f \, d\mu +\varepsilon y = \\
& = \varepsilon \int_S g' \, d\mu + (1- \varepsilon ) \int_S f\, d\mu + \int_T (f-g) \, d\mu + \int_T g \, d\mu = \\
& = \varepsilon \int_S g' \, d\mu + (1- \varepsilon ) \int_S f\, d\mu + (1- \varepsilon ) \int_{A \setminus S} (f-g) \, d\mu + \int_T g \, d\mu = \\
& = \varepsilon \int_S g' \, d\mu + (1- \varepsilon ) \int g\, d\mu - (1- \varepsilon ) \int_{A \setminus S} g \, d\mu + \int_T g \, d\mu = \\
& = \varepsilon \int_S g' \, d\mu + (1- \varepsilon ) \int_S g\, d\mu + \int_T g \, d\mu = \int_S \hat{g} \, d\mu + \int_T g \, d\mu = \int_{S \cup T} \tilde{g} \, d\mu
\end{align*}
This shows that $S \cup T$ can improve upon $(f,g)$ through $\left( \tilde{f} , \tilde{g} \right)$. Now select
\begin{equation}
\varepsilon = \frac{1 - \delta}{1- \mu (S)}
\end{equation}
then it is easily established that $0< \varepsilon <1$---since $0< \mu (S) < \delta <1$---and that
\[
\mu (S \cup T) = \mu(S) + (1- \varepsilon ) \mu (A \setminus S) = \delta .
\]
This completes the proof of the assertion.

\section{Proof of Theorem \ref{thm:CE}} \label{App-D}

\subsection{Proof of Theorem \ref{thm:CE}(a)}

Suppose that $(f,g)$ is a competitive equilibrium in the economy $\, \mathbb{E} = \langle \, (A, \Sigma ,\mu ), u, \mathcal{P} \rangle$ for the equilibrium price $p \in \mathbb{R}^{\ell}_+$ with $p \gg 0$.
\\
Suppose to the contrary that $(f,g)$ is \emph{not} a Core allocation in $\mathbb{E}$. Then there exists some non-negligible coalition $S \in \Sigma$ with $\mu (S)>0$ and a pair $(f',g')$ such that $g'(a) \in \mathcal{P}_a$ for all $a \in S$, $u_a (f'(a)) > u_a (f(a))$ for all $a \in S$ and $\int_S f' \, d\mu \leqslant \int_S g' \, d\mu$.
\\
From the assumption that $u_a (f'(a)) > u_a (f(a))$ for all $a \in S$, it follows from the consumer-producer problem for all $a \in S$ that $p \cdot f'(a) > p \cdot g(a) \geqslant p \cdot g'(a)$ for almost all $a \in S$.\footnote{Here, note that $p \cdot g(a) = \max p \cdot \mathcal{P}_a$ is the maximal income that agent $a$ can generate in this economy under price vector $p \gg 0$. This follows from the fact that $(f(a),g(a))$ is a solution to $a$'s consumer-producer problem. See also \citet[Theorem 2.4]{Gilles2016}. } Hence,
\[
p \cdot \int_S f' \, d\mu = \int_S p \cdot f'(a) \, d\mu (a) > \int_S p \cdot g'(a) \, d\mu (a) = p \cdot \int_S g' \, d\mu .
\]
This contradicts the assumed property that $\int_S f' \, d\mu \leqslant \int_S g' \, d\mu$. Hence, we have shown the assertion.

\subsection{Proof of Theorem \ref{thm:CE}(b)}

The proof of this assertion is a modification of the standard methodology introduced by \citet{Aumann1964} and elaborated in \citet[Section 2.1]{Hildenbrand1974}.
\\
Let $(f,g)$ be a Core allocation in the continuum economy $\mathbb{E}$. For every agent $a \in A$ we define the consumption bundles that are preferred to $f(a)$ by
\begin{equation}
\mathcal{U} (a) = \left\{ \left. \, x \in \mathbb{R}^{\ell}_+ \, \right| \, u_a (x) > u_a \left( f(a) \right) \, \right\}
\end{equation}
Also, define for $a \in A \colon$
\begin{equation}
\mathcal{Z} (a) = \left\{ \, \mathcal{U}(a) - \mathcal{P}_a \, \right\} \cup \{ 0 \} \subset \mathbb{R}^{\ell}
\end{equation}
Then $\mathcal{U} \colon A \twoheadrightarrow \mathbb{R}^{\ell}_+$ is a measurable correspondence and as a consequence, $\mathcal{Z} \colon A \twoheadrightarrow \mathbb{R}^{\ell}$ is a measurable correspondence as well. Now by standard arguments \citep[Theorem 3, page 62]{Hildenbrand1974}, $\int \mathcal{Z} \, d\mu$ is convex and non-empty. (In particular, $0 \in \int \mathcal{Z} \, d\mu$.)
\\
We explore the properties of the correspondence $\mathcal Z$ through a sequence of lemmas.
\begin{lemma} \label{Claim I}
$\int \mathcal{Z} \, d\mu \cap \mathbb{R}^{\ell}_- = \{ 0 \}$.
\end{lemma}
\begin{proof}
Suppose to the contrary that the claim does not hold. Then there exists an integrable selection $h \colon A \to \mathbb{R}^{\ell}$ such that $h(a) \in \mathcal{Z}(a)$ for almost all $a \in A$ and $\int h\, d\mu <0$.
\\
Next, let $S = \{ a \in A \mid h(a) \neq 0 \} \in \Sigma$. Since $\int h \, d\mu <0$ it holds that $\mu (S) >0$.
\\
Therefore, there exist integrable selections $f'$ in $\mathcal{U}$ and $g'$ in $\mathcal{P}$ such that $h = f'-g'$ with $h(a) =0$ if $a \notin S$. For every agent $a \in S$ we let
\begin{equation}
f''(a) = f' (a) - \frac{1}{\mu (S)} \int h \, d\mu \in \mathbb{R}^{\ell}_+
\end{equation}
In particular, since $\int h \, d\mu <0$, it holds that $f''(a) > f'(a)$ for every $a \in S$ and, therefore, by strict monotonicity of $a$'s preferences we conclude that
\[
u_a \left( f''(a) \right) > u_a \left( f'(a) \right) > u_a \left( f(a) \right)
\]
Also, we can easily establish that by definition
\begin{align*}
\int_S f'' \, d\mu & = \int_S f' \, d\mu - \int h \, d\mu = \\
& = \int_S f' \, d\mu - \int_S h \, d\mu = \int_S g' \, d\mu .
\end{align*}
Thus, the coalition $S$ is able to improve upon $(f,g)$ through the coalitional allocation $(f'',g')$. This contradicts the hypothesis that $(f,g)$ is a Core allocation. Thus, we have shown the assertion of Lemma \ref{Claim I}.
\end{proof}

\bigskip\noindent
Now by Lemma \ref{Claim I} we may apply Minkowski's Separation Theorem \citep[page 38]{Hildenbrand1974} to $\int \mathcal{Z} \, d\mu$ and $\mathbb{R}^{\ell}_-$. Therefore, there exists a normal vector $p >0$ such that $p \cdot \int \mathcal{Z} \, d\mu \geqslant 0$.
\\
Hence, $\inf p \cdot \mathcal{Z} (a) \geqslant 0$ for almost every $a \in A$. In particular, since $0 \in \mathcal{Z} (a)$, it follows that $\inf p \cdot \mathcal{Z}(a) =0$.
\\
So, we conclude now that for every $x \in \mathbb{R}^{\ell}_+ \colon$
\begin{equation}
\label{eq:ster}
u_a (x) > u_a(f(a)) \quad \mbox{implies} \quad p \cdot x \geqslant I(a,p) \equiv \sup p \cdot \mathcal{P}_a
\end{equation}
This now leads to the following claims:
\begin{lemma} \label{Claim II}
For almost every $a \in A \colon p \cdot f(a) = I(a,p) \geqslant 0$.
\end{lemma}
\begin{proof}
Let $a \in A$. By continuity of the utility function $u_a$ it follows from (\ref{eq:ster}) that $p \cdot f(a) \geqslant I(a,p)$. Now suppose that the assertion of Lemma \ref{Claim II} is not true, then there exists a coalition $S$ with $\mu (S)>0$ and for all $a \in S \colon p \cdot f(a) > I(a,p)$. Then
\[
p \cdot \int f \, d\mu > \int I(\cdot ,p) \, d\mu \geqslant p \cdot \int g \, d\mu .
\]
Since $p>0$, this property contradicts that $(f,g)$ is feasible. Hence, $p \cdot f(a) = I(a,p)$.
\\
Finally, by Axiom \ref{ax:basic}(iii), for $a \in A$ there exists some strictly positive vector $z_a  > 0$ in the convex hull of the production set $\mathcal P_a$, i.e., $z_a \in \mathrm{Conv} \, \mathcal P_a$. Hence, since $p>0$, we conclude that $I(a,p) = \sup p \cdot \mathcal{P}_a \geqslant p \cdot z_a \geqslant 0$, leading to the conclusion that Lemma \ref{Claim II} holds.
\end{proof}

\begin{lemma} \label{Claim III}
For almost every $a \in A \colon p \cdot g(a) = I(a,p)$.
\end{lemma}
\begin{proof}
Since $g (a) \in \mathcal{P}_a$ it follows that $p \cdot g(a) \leqslant I(a,p)$. Suppose that the assertion of Lemma \ref{Claim III} is not true, then there exists a coalition $S$ with $\mu (S)>0$ and $p \cdot g(a) < I(a,p)$. Since by Lemma \ref{Claim II}, $p \cdot f(a) = I(a,p)$ for almost all $a \in A$, we conclude that
\[
p \cdot \int f \, d\mu = \int I(a,p) \, d\mu (a) > p \cdot \int g \, d\mu .
\]
Hence, by $p > 0$, $\int f \, d\mu > \int g \, d\mu$, which contradicts the feasibility of $(f,g)$. Thus we have shown the validity of Lemma \ref{Claim III}.
\end{proof}

\begin{lemma} \label{Claim IV}
For almost every $a \in A \colon f(a)$ is $u_a$-maximal in the budget set
\[
B(a,p) = \left\{ \left. x \in \mathbb{R}^{\ell}_+ \, \right| \, p \cdot x \leqslant I(a,p) \, \right\} \neq \varnothing .
\]
\end{lemma}
\begin{proof}
Suppose $0 \leqslant p \cdot x < I(a,p)$ for $x \in \mathbb{R}^{\ell}_+$, then, by (\ref{eq:ster}), $u_a (x) \leqslant u_a (f(a))$.
\\
Suppose $p \cdot x = I(a,p)$ for $x \in \mathbb{R}^{\ell}_+$, then there exists a sequence $x^n \to x$ with $p \cdot x^n < I(a,p)$. Thus, we conclude that by continuity of $u_a$, $u_a(x) \leqslant u_a(f(a))$ as well.
\end{proof}

\begin{lemma} \label{Claim V}
	From strict monotonicity of all agents' preferences, it follows that $p \gg 0$ and $I(a,p) >0$ for almost every agent $a \in A$.
\end{lemma}
\begin{proof}
	If $p >0$ is such that at least one good $k \in \{ 1, \ldots ' \ell \}$ has a zero price $p_k=0$, it follows that $f'(a) = f(a) + e_k \in B(a,p)$ for almost all agents $a \in A$. By strict monotonicity of $u_a$ it then follows that $u_a(f'(a)) > u_a (f(a))$, which contradicts Lemma \ref{Claim IV}. This implies that $p \gg 0$.
	\\
	Furthermore, by Axiom \ref{ax:basic}(iii), for every $a \in A$ there exists some positive vector $z_a  > 0$ in the convex hull of the production set $\mathcal P_a$, i.e., $z_a \in \mathrm{Conv} \, \mathcal P_a$. Hence, since $p \gg 0$, we conclude that $I(a,p) = \sup p \cdot \mathcal{P}_a \geqslant p \cdot z_a > 0$. This proves Lemma \ref{Claim V}.
\end{proof}

\bigskip\noindent
The property that $p \gg 0$ (Lemma \ref{Claim V}) together with the assertions stated in Lemmas \ref{Claim III} and \ref{Claim IV} imply that the allocation $(f,g)$ is indeed a competitive equilibrium supported by the price vector $p$. This proves the assertion of Theorem \ref{thm:CE}(b).

\section{Proof of Theorem \ref{thm:Blocking}} \label{App-Blocking}

Let $\, \mathbb{E} = \langle \, (A, \Sigma ,\mu ), u, \mathcal{P} \rangle$ be some continuum economy, in which $(A, \Sigma ,\mu )$ is a complete atomless probability space for which the assumptions introduced in Axiom \ref{ax:basic} and Definition \ref{def:economy} hold and all utility functions $u_a$, $ a \in A$, are strictly monotone. As stated in the assertion, we also suppose that all agents in $\mathbb{E}$ have productive abilities that are subject to Increasing Returns to Specialisation (IRSpec), i.e., for almost every agent $a \in A \colon \mathcal Q (a) \subset \mathcal P_a \subset \mathrm{Conv} \, \mathcal Q (a) - \mathbb R^{\ell}_+$, where
\begin{equation}
\mathcal{Q} (a) = \left\{ Q^1 (a) e_1 -y^1 (a), \ldots ,Q^{\ell} (a) e_{\ell} -y^{\ell} (a) \right\} .
\end{equation}
for well-selected output levels $Q^k (a) >0$ and input vectors $y^k (a) \in \mathbb{R}^{\ell}_+$ for $k \in \{ 1, \ldots , \ell \}$.
\\
We introduce some auxiliary notation. In particular, we define for every $a \in A \colon$
\begin{equation}
	\overline{\mathcal Q} (a) = \mathrm{Conv} \, \mathcal Q (a)
\end{equation}
Now we can prove the following assertion:

\begin{lemma} \label{lem:Q-measurable}
	The correspondences $\mathcal Q \colon A \to 2^{\mathbb R^{\ell}}$ and $\overline{\mathcal Q} \colon A \to 2^{\mathbb R^{\ell}}$ have a measurable graph.
\end{lemma}
\begin{proof}
	First, as assumed in Definition \ref{def:economy}, the correspondence $\mathcal P \colon A \to 2^{\mathbb R^{\ell}}$ has a measurable graph and is closed-valued. From the IRSpec property it follows that for every $a \in A$ and for every $k \in \{ 1, \ldots , \ell \}$ the maximisation problem
	\[
	\max \phi_k (x) = e_k \cdot x \quad \mbox{such that } x \in \mathcal P_a
	\]
	has a unique solution given by $\phi_k (z^a_k) = Q^k (a)$ with $z^a_k = Q^k (a) e_k - y^k (a) \in \mathcal Q (a)$.
	\\
	Proposition 3 in \citet[page 60]{Hildenbrand1974} now implies that $a \mapsto z^a_k$ is a measurable function on the complete probability space $(A, \Sigma , \mu )$. This, in turn, implies that
	\[
	a \mapsto \mathcal Q (a) = \left\{ Q^1 (a) e_1 -y^1 (a), \ldots ,Q^{\ell} (a) e_{\ell} -y^{\ell} (a) \right\} = \bigcup^{\ell}_{k=1} \left\{ z^a_k \right\}
	\]
	has a measurable graph on $(A, \Sigma , \mu )$.
	\\
	Finally, by the Corollary of Proposition 3 in \citet[page 60]{Hildenbrand1974}, this implies that the correspondence $\overline{\mathcal Q} \colon A \to 2^{\mathbb R^{\ell}}$ that assigns to every $a \in A$ the convex hull $\mathrm{Conv} \, \mathcal Q (a)$ of the finite set $\mathcal Q (a)$ has a measurable graph.
\end{proof}

\bigskip\noindent
From Theorem \ref{thm:CE} it follows that $\mathcal C( \mathbb E) = \mathcal W (\mathbb E) \subset \mathcal C^S ( \mathbb E)$. It is left to show that $\mathcal C^S ( \mathbb E) \subset \mathcal C( \mathbb E)$.
\\[2ex]
Now, let $(f,g) \notin \mathcal C( \mathbb E)$ be some non-Core allocation in $\mathbb{E}$. Hence, there exists some $S \in \Sigma$ with $\mu (S)>0$ and a coalitional allocation $(f',g')$ with $f' \colon S \to \mathbb{R}^{\ell}_+$ and $g' \colon S \to \mathbb{R}^{\ell}$ such that
\begin{numm}
\item $u_a(f'(a)) > u_a(f(a)) \equiv \bar{u}_a$ for every $a \in S$;

\item $g'(a) \in \mathcal{P}_a$ for every $a \in S$, and;

\item $\int_S f' \, d\mu = \int_S g' \, d\mu$.
\end{numm}
Next, since $\mathcal Q$ has a measurable graph by Lemma \ref{lem:Q-measurable}, we may define
\begin{equation}
	Q_S = \int_S \mathcal Q (a) \, d\mu (a)
\end{equation}
We note that, by Theorem 4 in \citet[page 64]{Hildenbrand1974}, from the atomlessness of $(A, \Sigma , \mu )$ it follows that
\begin{equation}
	Q_S = \int_S \mathrm{Conv} \, \mathcal Q (a) \, d\mu (a) = \int_S \overline{\mathcal Q} (a) \, d\mu (a) \neq \varnothing
\end{equation}
is a closed and convex set.\footnote{Nonemptiness of $Q_S$ follows from the integrably boundedness of the correspondence $\mathcal P$ and, therefore, of $\mathcal Q$. This implies that all measurable selections in $\mathcal Q$ are integrable.}
\\
We now claim the following property:
\begin{lemma} \label{lem:Q-delta}
	Without loss of generality, we may assume that $\int_S g' \, d\mu \in Q_S$.
\end{lemma}
\begin{proof}
Suppose to the contrary that $\int_S g' \, d\mu \notin Q_S$.
\\
Then there exists a subcoalition $T \subset S$ with $\mu (T) >0$ such that $g' (a) \notin \overline{\mathcal Q} (a)$ for every $a \in T$.
\\
Note first that for every $a \in A$ the convex set $\overline{\mathcal Q}(a) \subset \mathbb R^{\ell}$ is an $(\ell -1)$-dimensional subspace of $\mathbb R^{\ell}$. Therefore, since for every $a \in T \colon \mathcal P_a \subset \overline{\mathcal Q} (a) - \mathbb R^{\ell}_+$, $g' (a) \in \mathcal P_a$ and $\overline{\mathcal Q}(a)$ is $(\ell-1)$-dimensional, there exists some vector $z(a) > 0$ such that
\[
g' (a) + z(a) \in \overline{\mathcal Q} (a) .
\]
Without loss of generality, we may assume that $z \colon T \to \mathbb R^{\ell}_+$ is measurable.\footnote{This can be shown through an appropriate application of Proposition 3 in \citet[page 60]{Hildenbrand1974}.} Since the production correspondence $\mathcal P$ is integrably bounded from above by Definition \ref{def:economy}(iii) it follows, therefore, that the mapping $z$ is integrable with $\mathbf{z} = \int_T z(a) \, d\mu (a) >0$. Then 
\[
\int_S g' \, d\mu + \mathbf z = \int_T \left( g'(a) + z(a) \, \right) \, d\mu (a) + \int_{S \setminus T} g' \, d\mu \in \int_S \overline{\mathcal Q} (a) \, d\mu (a) = Q_S .
\]
Then there exists some integrable selection $\hat g \colon S \to \mathbb R^{\ell}$ with $\hat g(a) \in Q(a)$ such that
\[
\int_S \hat g \, d\mu = \int_S g' \, d\mu + \mathbf z \in Q_S .
\]
Now, the coalition $S$ improves upon $(f,g)$ through $(f', \hat g )$. Indeed,
\[
\int_S f' \, d\mu \leqslant \int_S g' \, d\mu <  \int_S g' \, d\mu + \mathbf z = \int_S \hat g \, d\mu .
\]
This shows the assertion of the lemma.
\end{proof}

\bigskip\noindent
From Lemma \ref{lem:Q-delta}, we may assume that $\int_S g' \, d\mu \in Q_S$. Therefore, there exists some integrable selection $g'' \colon S \to \mathbb R^{\ell}$ with $g'' (a) \in \mathcal Q (a)$ such that
\[
\int_S g'' \, d\mu = \int_S g' \, d\mu .
\]
We conclude that $S$ improves upon $(f,g)$ through $(f',g'') \colon u_a (f'(a)) > u_a(f(a))$ for all $a \in S$ and
\[
\int_S f' \, d\mu \leqslant  \int_S g' \, d\mu =  \int_S g'' \, d\mu
\]
Hence, $(f,g) \notin \mathcal C^S ( \mathbb E)$, showing the assertion of Theorem \ref{thm:Blocking}.

\end{document}